\documentclass[11pt]{amsart}
\usepackage{amsmath}
\usepackage[dvips]{epsfig}
\usepackage{slashed}
\usepackage{accents}
\usepackage{color}
\usepackage{leftidx}
\usepackage{verbatim}

\theoremstyle{plain}
\newtheorem{theorem}{Theorem}

\newtheorem{lemma}[theorem]{Lemma}

\DeclareMathOperator{\ric}{Ric}
\DeclareMathOperator{\diam}{diam}
\DeclareMathOperator{\tr}{tr}

\begin{document}

\title[Nonexistence of Extremal de Sitter Black Rings] {Nonexistence of Extremal de Sitter Black Rings}

\author[Khuri]{Marcus Khuri}
\address{Department of Mathematics\\
Stony Brook University\\
Stony Brook, NY 11794, USA}
\email{khuri@math.sunysb.edu}

\author[Woolgar]{Eric Woolgar}
\address{Department of Mathematical and Statistical Sciences\\
University of Alberta\\
Edmonton, AB, Canada T6G 2G1}
\email{ewoolgar@ualberta.ca}

\thanks{M. Khuri acknowledges the support of NSF Grant DMS-1708798. E. Woolgar was supported by a Discovery Grant RGPIN 203614 from the Natural Sciences and Engineering Research Council.}

\begin{abstract}
We show that near-horizon geometries in the presence of a positive cosmological constant cannot exist with ring topology. In particular, de Sitter black rings with vanishing surface gravity do not exist. Our result relies on a known mathematical theorem which is a straightforward consequence of a type of energy condition for a modified Ricci tensor, similar to the curvature-dimension conditions for the $m$-Bakry-\'Emery-Ricci tensor.
\end{abstract}
\maketitle

\setcounter{equation}{0}
\setcounter{section}{1}

\noindent The discovery of the Emparan-Reall singly spinning black ring \cite{EmparanReall} and the
Pomeransky-Sen'kov doubly spinning black ring \cite{PomeranskySenkov} have played an important role in the theory of higher dimensional black holes. A basic open question has been whether such
solutions may be generalized to the cosmological setting? There has been relatively little progress in constructing such solutions. One potential reason for this is that dimensional reduction of the Einstein equations with nonzero cosmological constant $\Lambda\neq 0$ fails to yield a sigma model structure \cite{HollandsIshibashi}, in contrast to the $\Lambda=0$ case. Thus, standard solution generating techniques do not apply in this context. In fact, it turns out that in the extremal case these solutions cannot exist for topological reasons.

In order to state the main result we first introduce a bit of terminology. Consider the Einstein equations
\begin{equation}
\label{eq1}
R_{\mu\nu}=\Lambda g_{\mu\nu}+T_{\mu\nu}-\frac{1}{n}g^{\rho\sigma}T_{\rho\sigma}g_{\mu\nu},
\end{equation}
where $T$ is the energy-momentum tensor and $n=D-2$, with $D$ denoting the dimension of the spacetime. A \textit{de Sitter black ring} is taken to mean a regular spacetime satisfying the Einstein equations with $\Lambda >0$, and having a Killing horizon with cross-section topology $S^1\times \Sigma$ where $\Sigma$ is an arbitrary compact manifold of dimension $D-3$. Recall that a Killing horizon is a null hypersurface defined by the vanishing in norm of a Killing field $V$, which is normal to the horizon. These come naturally equipped with a notion of surface gravity $\kappa$ defined through the equation
\begin{equation}
\label{eq2}
\nabla_{V}V=\kappa V
\end{equation}
on the horizon.

\begin{theorem}\label{thm1}
There do not exist de Sitter black rings with zero surface gravity and matter fields satisfying the energy condition \eqref{energy}. In particular this conclusion holds in vacuum.
\end{theorem}

It should be noted that this theorem does not require any symmetry hypotheses beyond the Killing horizon assumption. Previous nonexistence results have been established in \cite{GroverGutowskiSabra,KunduriLuciettiReall}, in which regular supersymmetric anti-de Sitter black rings have been found not to exist in 5D minimal gauged supergravity. In addition, a general physical argument against the existence of supersymmetric AdS black rings may be found in \cite{CaldarelliEmparanRodriguez}. Since supersymmetry is not compatible with $\Lambda>0$, these results imply nonexistence of such ring solutions in the presence of a nonvanishing cosmological constant. On the other hand, there is strong evidence for the existence of non-extremal (A) dS black rings, and in fact a perturbative construction of them has been given \cite{CaldarelliEmparanRodriguez}.

The proof of Theorem \ref{thm1} is based on what we view as an important unrecognized relationship between near-horizon geometries and the mathematical notion of $m$-quasi-Einstein metrics. When a degenerate ($\kappa=0$) Killing horizon is present, Gaussian
null coordinates \cite{KunduriLucietti} may be introduced in a neighborhood of the horizon so that the spacetime metric takes the form
\begin{equation}
\label{eq3}
g = 2 dv \left(dr +\frac{1}{2}r^2 F(r,x) dv +rh_a(r,x) dx^a\right)
 + \gamma_{ab}(r,x) dx^a dx^b,
\end{equation}
where $V=\partial_{v}$, the horizon is located at $r=0$, and $\gamma$ represents the
induced metric on the horizon cross-section $\mathcal{H}$. By taking the near-horizon limit
$v\rightarrow \frac{v}{\varepsilon}$,
$r\rightarrow \varepsilon r$, and
$\varepsilon\rightarrow 0$,
we obtain the near-horizon geometry
\begin{equation}
\label{eq4}
g_{NH} = 2 dv \left(dr +\frac{1}{2}r^2 F(x) dv +rh_a(x) dx^a\right)
 + \gamma_{ab}(x) dx^a dx^b
\end{equation}
which is determined by the near-horizon data $(\gamma_{ab},h_a, F)$
living on $\mathcal{H}$. These must satisfy the near-horizon geometry equations
\begin{equation}
\label{eq5}
\begin{split}
R_{ab}=&\, \frac{1}{2}h_a h_b -\nabla_{(a}h_{b)}+\Lambda\gamma_{ab}+P_{ab},\\
F=&\, \frac{1}{2}|h|^2-\frac{1}{2}\nabla_{a}h^a+\Lambda +E,
\end{split}
\end{equation}
where
\begin{equation}
\label{eq6}
\begin{split}
P_{ab}=&\, T_{ab} -\frac{1}{n}\left ( \tr_{g_{NH}} T_{NH}\right ) \gamma_{ab}
= T_{ab}-\frac{1}{n}\left(\gamma^{cd}T_{cd}+2T_{+-}\right)\gamma_{ab},\\
E=&\,-\left(\frac{n-2}{n}\right)T_{+-}+\frac{1}{n}\gamma^{ab}T_{ab}.
\end{split}
\end{equation}
Here $T_{+-}$ denotes the $v-r$ component in the near-horizon limit of the energy-momentum tensor, parameterized in \cite{KunduriLucietti} as
\begin{equation}
\label{eq7}
T_{NH}=2T_{+-}dr dv +2r\left ( \beta_a+T_{+-}h_a\right ) dx^a dv +r^2\left ( \alpha + T_{+-}F\right )dv^2+T_{ab}dx^adx^b\ .
\end{equation}
The relevant energy condition states that the symmetric matrix $P$ is nonnegative definite, that is
\begin{equation}\label{energy}
P\geq 0.
\end{equation}

On the other hand, the \emph{modified Ricci tensor} of \cite{Limoncu} on $\mathcal{H}$ is defined by
\begin{equation}
\label{eq9}
\ric_{X}^{m}=\ric+\frac{1}{2}L_{X}\gamma-\frac{1}{m}X\otimes X,
\end{equation}
where $X$ is a vector field/1-form (we use the same notation for both) on $\mathcal{H}$ and $L$ denotes Lie differentiation. We note that in the special case where $X=\nabla f$, the modified Ricci tensor becomes the $m$-Bakry-\'Emery-Ricci tensor $\ric_m^f$ (or sometimes $\ric_{n+m}^f$, depending on convention); here however, we do not assume $X$ to be a gradient vector field. We will refer to the metric $\gamma$ as an \emph{$m$-quasi-Einstein metric}\footnote
{In contrast to our usage, some authors reserve this term for the $X=\nabla f$ case.}
if there exists $X$ and a constant $\lambda$ such that
\begin{equation}
\label{eq10}
\ric_{X}^m=\lambda \gamma.
\end{equation}
Thus we find that by setting $X=h$ and $m=2$, a vacuum near-horizon geometry defines an
$m$-quasi-Einstein metric.

It turns out that the theory of $m$-quasi-Einstein metrics parallels that of Riemannian
geometry. By this we mean that many standard results of Riemannian geometry have extensions to the setting of $m$-quasi-Einstein metrics. One of the most important is that of Myers theorem, which asserts that a positive lower bound for the Ricci curvature of a complete manifold implies a corresponding upper bound for the diameter \cite{Petersen}.
The version of this theorem for $m$-quasi-Einstein metrics \cite{Limoncu} is what we will use to prove Theorem \ref{thm1}.

\begin{theorem}[{\cite[Theorem 1.2]{Limoncu}}]\label{thm2}
Let $(\mathcal{H},\gamma)$ be a complete Riemannian $n$-manifold satisfying
\begin{equation*}
\ric_{X}^m\geq\lambda\gamma>0,\quad\quad\quad |X|\leq C\ .
\end{equation*}
Then
\begin{equation*}
\diam(\mathcal{H})\leq \frac{\pi}{\lambda}\left ( \frac{C}{\sqrt{2}}+\sqrt{ \frac{C^2}{2}+(n-1)\lambda}\right )  .
\end{equation*}
\end{theorem}

Observe that the only additional hypothesis beyond those of the classical theorem
is the requirement that the vector field $X$ be uniformly bounded. On a compact manifold (the horizon cross-section) this is automatically satisfied. We are now in a position to establish Theorem \ref{thm1} using a standard topological consequence of this result.

\begin{proof}[Proof of Theorem \ref{thm1}] Suppose that a de Sitter black ring exists with zero surface gravity and matter fields satisfying the energy condition \eqref{energy}. The associated near-horizon geometry then satisfies the hypotheses of Theorem \ref{thm2} with $h=X$, $\lambda=\Lambda$, and $m=2$. By pulling back the relevant geometric quantities to the universal covering manifold, we find that the diameter of the universal cover must be finite. This implies that the base $\mathcal{H}$ must have finite fundamental group. However $\pi_{1}(\mathcal{H})=\pi_{1}(S^1\times\Sigma)=\mathbb{Z}\times \pi_{1}(\Sigma)$ which is infinite, a contradiction.
\end{proof}

This line of argument may be interpreted as follows in the realm of near-horizon geometries, and is of independent interest.

\begin{theorem}\label{theorem3}
There do not exist de Sitter near-horizon geometries of ring type with matter fields
satisfying the energy condition \eqref{energy}. In particular this conclusion holds in vacuum.
\end{theorem}

Finally, we interpret the energy condition \eqref{energy} in terms of perfect fluids. Perfect fluids have energy-momentum tensors of the form
\begin{equation}
\label{eq11}
T = \left ( \rho + p \right ) u\otimes u +p g,
\end{equation}
where $\rho$ is the non-gravitational energy density, $p$ is the pressure, and $u$ is a unit timelike vector. Then $\tr_{g_{NH}} T_{NH}= -\rho+(D-1)p = -\rho + (n+1)p$ and we obtain
\begin{equation}\label{eq12}
P_{ab}= T_{ab}-\frac{1}{n}\left ( \tr_{g_{NH}} T_{NH}\right ) \gamma_{ab}
= (\rho +p) u_{a} u_{b} +\frac{1}{n}(\rho -p)\gamma_{ab},
\end{equation}
It follows that \eqref{energy} is assured to hold for all (timelike) $u$ if $\rho\ge |p|$. But $\rho\ge |p|$ is the \emph{dominant energy condition} for perfect fluids.

\begin{lemma}\label{lemma4}
If matter is described by a perfect fluid obeying the dominant energy condition, then the energy condition \eqref{energy} holds.
\end{lemma}

We close by noting that, while the brevity of our argument arises in part from the fact that the proof of Theorem \ref{thm2} is already available in the literature, the proof of Theorem \ref{thm2} in \cite{Limoncu} is itself brief and direct, and uses only familiar methods in geodesic geometry. This to us suggests that the further exploitation of geodesic geometry of quasi-Einstein metrics may yield further results in horizon geometry with little fuss or effort.

\textbf{Acknowledgements.}
The authors would like to thank Roberto Emparan for insightful comments,
Greg Galloway for bringing reference \cite{Limoncu} to their attention, and Hari Kunduri for useful discussions. The authors thank the Erwin Schr\"odinger International Institute for Mathematics and Physics and the organizers of its ``Geometry and Relativity'' program, where this paper was conceived and written.

\end{document}